\documentclass[aps,twocolumn,superscriptaddress,nofootinbib,a4paper,longbibliography]{revtex4-2}

\usepackage[latin9]{inputenc}
\usepackage{times}
\usepackage{graphicx}
\usepackage{amsfonts}
\usepackage{amsmath}
\usepackage{amssymb}
\usepackage{dsfont}
\usepackage{amsthm}
\usepackage{xcolor}
\usepackage[colorlinks=true,linkcolor=blue,citecolor=blue,urlcolor=blue]{hyperref}
\usepackage{comment}
\usepackage{braket}
\usepackage{physics}
\usepackage{multirow}
\usepackage{thmtools}
\usepackage{thm-restate}
\usepackage{cleveref}
\usepackage{enumitem}

\makeatletter

\theoremstyle{definition}
\newtheorem{theorem}{Theorem}
\newtheorem{lemma}[theorem]{Lemma}  

\newtheorem{proposition}[theorem]{Proposition}

\makeatother

\newcommand{\id}{\mathbf{1}}

\newcommand{\state}[2]{\langle #2 \rangle_{#1}}

\begin{document}

\title{Self-testing with untrusted random number generators}

\author{Mois\'es Bermejo Mor\'an}
\email{moises@hku.hk}
\affiliation{Department of Computer Science, School of Computing and Data Science, The University of Hong Kong, Pokfulam Road, Hong Kong, China}
\author{Ravishankar Ramanathan}
\email{ravi@cs.hku.hk}
\affiliation{Department of Computer Science, School of Computing and Data Science, The University of Hong Kong, Pokfulam Road, Hong Kong, China}

\begin{abstract}
Self-testing--the attractive possibility to infer the underlying physics of a quantum device in a black-box scenario--has gained increased traction in recent years, with applications to device-independent quantum information processing. Thus far, self-testing has been done under the assumption that the settings for the requisite Bell test are chosen freely and independently of the device tested in the experiment. That is, the random number generator used to generate the settings has been assumed to have no correlations with the device tested. Here, we extend self-testing protocols beyond the independence assumption. Surprisingly, we show that all pure bipartite partially entangled states can be self-tested provided that the random number generator obeys a residual randomness constraint strictly weaker than the independence assumption. This in itself provides a semi-device-independent certification of independence between the randomness source and the device.
\end{abstract}

\maketitle

\textit{Introduction.-} 
The foundational phenomenon of Bell nonlocality has given rise to the possibility of device-independent quantum cryptography, specifically the possibility of obtaining secure key \cite{ABG07} and private, fully random bits for use in cryptographic applications \cite{PAM+10}. An attractive feature of these applications is the fact that the users do not need to trust their devices and can instead directly verify the security by means of simple statistical tests on the input-output behavior of the device, specifically checking the violation of a Bell inequality or the success probability in a nonlocal game. This violation certifies, even in a fully adversarial scenario where an eavesdropper has tampered with the device or holds additional systems correlated with the device, that the outputs of the device could not be perfectly predicted by the eavesdropper. 

A particularly important property of certain Bell inequalities is self-testing, which allows the users to fully determine the underlying quantum state of the device and the measurement performed (up to local isometries), based solely on the input-output behaviour. A famous example is the Clauser-Horne-Shimony-Holt (CHSH) inequality \cite{CHSH69}, whose maximal quantum value self-tests the state of the device to be an EPR pair $|\phi_{+} \rangle = \left(|00 \rangle + |11 \rangle \right)/\sqrt{2}$ as well as particular Pauli measurements performed on the state \cite{MY04}. Other well-known examples include the family of tilted CHSH inequalities \cite{YN12}, the magic square game \cite{WBMS16}, the tripartite Mermin inequality \cite{McK11}, families of $n$-player binary XOR games \cite{MYS12}  as well as a host of other interesting examples with foundational and practical implications \cite{RMPBellnonlocality, SB20}. Self-testing has grown into a highly interesting research field with profound implications for cryptography \cite{ABG07} and quantum computational complexity \cite{NV17}.

The success of these self-testing protocols and practical applications heavily relies on the \emph{measurement independence} assumption: the choice of measurement settings in the Bell scenario is random and independent of the device. While this assumption can be naturally justified in an ideal scenario by the free will of the people performing the experiment to choose the settings, in realistic applications the settings are chosen with random number generators (RNG). 
It is therefore against the spirit of device-independence to trust that these sources will be perfectly random and uncorrelated with the device. 
Indeed, arbitrary correlations between the choice of measurement settings and the state of the device permit to describe all measurement statistics with local hidden variables \cite{B88, Hall11}.

Additional assumptions on the source are therefore necessary if one expects to extend device-independent protocols to Bell scenarios with untrusted sources of randomness. Under the assumption that the outputs of the random number generators cannot be perfectly predicted with the device, we show that it is possible to self-test any partially entangled bipartite pure state. Our assumption is strictly weaker than measurement independence and is satisfied, in particular, in the device-independent protocols for randomness amplification considered in \cite{WBG18, SFR25}.

We finally expose a fundamental distinction between purely probabilistic Bell tests and Hardy-type tests that rely on impossibilities in realistic scenarios: Bell values fail to self-test under arbitrarily weak measurement dependence, whereas Hardy self-tests succeed under arbitrarily weak independence.

\textit{Bell scenarios with untrusted sources.-} A bipartite Bell scenario is an experiment with local input spaces $X$ and $Y$ and output spaces $A$ and $B$, which possibly depend on some hidden variable $\lambda \in \Lambda$. The probabilities of joint events $a,b,x,y$,
\begin{equation}
    p(abxy) = \sum_\lambda p(\lambda) p(xy|\lambda) p(ab|xy\lambda) \, ,
\end{equation}
are obtained by averaging over all (unobserved) hidden variables $\lambda \in \Lambda$.
Under the assumption of \emph{measurement independence}, the input variables do not depend on the hidden variable $p(xy|\lambda) = p(xy)$, and are often assumed to be uniform $p(xy) = 1/|XY|$. Therefore, all information is contained in the conditional distributions $p(ab|xy)$. These are \emph{classical} when they decompose into local distributions $p(ab|xy) = \sum_\lambda p(\lambda) p(a|x\lambda) p(b|y\lambda)$. \emph{Quantum} distributions are obtained as expectation values of local measurement effects under some shared state. That is,
\begin{equation}
p(ab|xy) = \state{}{A_{a|x} \otimes B_{b|y}} = \tr ( \rho A_{a|x} \otimes B_{b|y} )\, ,
\end{equation}
where $\rho$ is a state in a bipartite Hilbert space $H_{A} \otimes H_{B}$, $A_x = (A_{a|x})_{a\in A}$ is a projective measurement on $H_A$ (i.e. $\sum_a A_{a|x} = \id_A$ and $A_{a|x}^2 = A_{a|x}$) and $B_y = (B_{b|y})_{b \in B}$ a projective measurement on $H_B$ for each $x,y$. 

We consider instead a realistic scenario that violates the measurement independence assumption, in which the input variables $x\in X$ and $y \in Y$ are chosen with random number generators $S$ and $T$. In this scenario, depicted in Figure~\ref{fig:causal-diagram}, the joint events have probabilities
\begin{equation}
    p(stabxy) = \sum_\lambda p(\lambda) p(xy|st) p(stab|xy\lambda) \, ,
\end{equation}
where $p(stab|xy\lambda) = p(st|\lambda)p(ab|stxy\lambda)$ can be understood as conditional distributions $p_{st}(ab|xy\lambda)$ subnormalized to $p(st|\lambda)$. Indeed, this setting can be alternatively visualized as a Bell scenario with two additional parties with empty input spaces and output spaces $S$ and $T$, wired to $X$ and $Y$. A distribution is \emph{classical} when it decomposes into local ones, $p(stab|xy) = \sum_{\lambda}p(\lambda) p(st|\lambda)p(a|x\lambda)p(b|y\lambda)$. \emph{Quantum} distributions are those obtained in the four-partite Bell scenario with classical-quantum states $\rho = \sum_{st}|st\rangle \langle st| \otimes \rho_{st}$. That is, with expectation values of local measurement effects on the device under (subnormalized) states of the device classically correlated with the sources,
\begin{equation}
p(stab|xy) = \state{st}{A_{a|x} \otimes B_{b|y}} = \tr(\rho_{st} A_{a|x} \otimes B_{b|y}) \, . 
\end{equation}
 We consider device-independent protocols in which we only have access to \emph{observable} events of the form $p(stab|st)$, in which the settings are chosen from the sources.

We say that the source has \emph{residual randomness} if
\begin{equation}
p(st | abxy) > 0
\label{eq:source}
\end{equation}
for each $s,t,a,b,x,y$. Operationally, this imposes that measurement statistics on the device cannot perfectly discriminate the outputs of the source. In quantum theory, this means that the measurement effects $A_{a|x} \otimes B_{b|y}$ cannot discriminate the states of the device $\rho_{st}$ classically correlated with the source. The assumption that the source has randomness even conditioned upon the quantum state of the device also appears in the Santha-Vazirani condition for boxes \cite{WBG18} or quantum Santha-Vazirani condition \cite{SFR25} considered in the context of randomness amplification.
Notice that the residual randomness assumption is strictly stronger than the condition $p(st|xy) > 0$ often considered in measurement dependence scenarios \cite{PRBLG14}.

An immediate consequence of the residual randomness assumption is that \emph{the impossibility of events is independent of the source}. This condition is strictly weaker than independence for all events, but is enough to extend quantum information protocols based on impossibility beyond the measurement independence scenario. Examples of these protocols include all-versus-nothing proofs of nonlocality \cite{M90} and pseudo-telepathy games \cite{BBT05}, which require stronger forms of nonlocality~\cite{LCZ24} with practical consequences in quantum computing~\cite{BGK18}, quantum complexity theory~\cite{JNV21} and randomness amplification~\cite{FM18}. 

\begin{lemma}
    Let $p(stab|xy)$ be a distribution obtained in a Bell experiment with untrusted sources with residual randomness. Then $p(stab|xy) = 0$ for some $s,t$ if and only if $p(s't'ab|xy) = 0$ for every $s',t'$.
    \label{lemma:impossible}
\end{lemma}

\begin{proof}
Imagine $0 = p(stab|xy) = p(st|abxy) p(ab|xy)$. The residual randomness assumption $p(st|abxy) > 0$ implies that $0 = p(ab|xy) = \sum_{s't'} p(s't'ab|xy)$. Since each $p(s't'ab|xy)$ is non-negative, $p(s't'ab|xy) = 0$ for each $s',t'$.
\end{proof}

\begin{figure}
    \centering
    \includegraphics[width=0.4\linewidth]{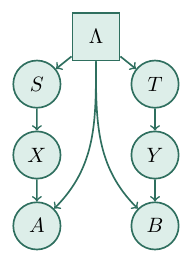}
    \includegraphics[]{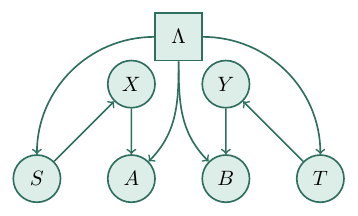}
    \caption{ \textbf{Bell scenario with untrusted sources.} The sources $S$ and $T$ are chosen to pick the setting from the input space $X$ and $Y$, which produce outputs from spaces $A$ and $B$ with no explicit dependence on the sources. The hidden variable $\Lambda$ influences $S, T, A, B$. We can understand $S$ and $T$ as the output spaces of two additional parties in the Bell experiments with no inputs. We consider classical-quantum states $\rho = \sum_{st} |st \rangle \langle st | \otimes \rho_{st}$, which produce probabilities $p(stab|xy) = \state{st}{A_{a|x} \otimes B_{b|y}} = \tr (\rho_{st} A_{a|x} \otimes B_{b|y})$ satisfying $p(st|abxy) > 0$. The protocols only have access to the \emph{observed} probabilities, which have the form $p(stab|st)$.}
    \label{fig:causal-diagram}
\end{figure}

\textit{Quantum distributions.-}
    Quantum distributions are obtained as expectation values of local measurement effects under (subnormalized) states classically correlated with the source, $p(stab|xy) = \state{st}{A_{a|x} \otimes B_{b|y}}$. These local measurement effects can be relaxed with global commuting ones, $p(stab|xy) = \state{st}{A_{a|x} B_{b|y}}$ with $[A_{a|x}, B_{b|y}] = 0$, which enables an algebraic characterization of quantum distributions more amenable to computational implementations \cite{NPA1, NPA2}. Considering unspecified states and unspecified measurement effects amounts to considering arbitrary unital positive linear functionals over the $*$-algebra $\mathcal A = \mathds C \langle A_{a|x}, B_{b|y}; \mathcal R \rangle$, where $\mathcal R$ is the collection of constraints $A_{a|x}^2 = A_{a|x}$, $B_{b|y}^2 = B_{b|y}$, $\sum_{a} A_{a|x} = \id$ and $\sum_{b} B_{b|y} = \id $ for each $a,b,x,y$. That is, the effects are treated as hermitian non-commutative variables subject to algebraic relations.
    A distribution $p(stab|xy)$ is quantum if and only if the following problem is feasible: 
    \begin{align}
        p(stab|xy) = \, &  L_{st}(A_{a|x} B_{b|y}) \, , \\
        & L_{st} \in \mathcal A^* \, , \nonumber \\ 
        & L_{st}(\id) = p(st) \, , \nonumber \\
        & L_{st} \geq 0 \, . \nonumber
    \end{align}
    Here, $\mathcal A^*$ is the space of linear functionals on $\mathcal A$, $\id$ is the identity in $\mathcal A$ and $L_{st} \geq 0$ means that $L_{st}(a^*a) \geq 0$ for each $a \in \mathcal A$.
    In particular, for each finite subspace $V \subset \mathcal A$
    \begin{align}
        p(stab|xy) = \, &  L_{st}(A_{a|x} B_{b|y}) \, , \label{eq:quantum_source_relax}\\
        & L_{st} \in V^* \, , \nonumber \\ 
        & L_{st}(\id) = p(st) \, , \nonumber \\
        & L_{st} \geq 0 \, , \nonumber
    \end{align}
     must be feasible. This is a semidefinite program, which can be solved in polynomial time in the dimension of $V$ \cite{K84, NN94}. If any such program is not feasible, $p(stab|xy)$ is not compatible with quantum theory. Moreover, for each $0 < l \leq u < 1$, maximizing Bell expressions under the linear functionals in Eq.~\eqref{eq:quantum_source_relax} that satisfy the additional linear constraints
    \begin{equation}
        l  \sum_{s't'} p(s't'ab|xy) \leq p(stab|xy) \leq u \sum_{s't'} p(s't'ab|xy)
    \end{equation}
    provides a convergent sequence of upper bounds to the maximal quantum value with untrusted sources whose residual randomness satisfies $ l \leq p(st|abxy) \leq u $.
    
\textit{Self-testing qubit states with untrusted sources.-} 
From now on we consider the Bell scenario depicted in Figure~\ref{fig:causal-diagram}, where the sources satisfy the residual randomness condition $p(st|abxy)>0$ in Eq.~\eqref{eq:source}. We show that for each partially entangled bipartite pure state $|\psi \rangle$ there exist conditions on the distribution $p(stab|xy)$ obtained in a suitable Bell experiment with untrusted sources that can only be satisfied when each $\rho_{st} = p(st) |\psi \rangle \langle \psi |$ up to local isometries. In such a case, the device is certified to be independent from the source, $\rho = \sum_{st} p(st) |st \rangle \langle st | \otimes | \psi \rangle \langle \psi |$.

\begin{proposition}
\label{prop:selftesting-2qub-MI}
Consider the following conditions on the distribution $p(stab|xy)$ obtained in a Bell scenario with two inputs and two outputs, with untrusted sources: 
\begin{align}
p(0101|01) & = 0 \, , \label{eq:HT-zero-1}\\
p(1010|10) & = 0 \, , \label{eq:HT-zero-2}\\
p(1100|11) & = 0 \, , \label{eq:HT-zero-3}\\
p(0000|00) & + wp(0011|00) = p(00)q(w) \, ,\label{eq:HT-violation}
\end{align}
where $q(w)$ is the (normalized) maximal quantum value compatible with the zeroes. These conditions can only be satisfied when each $\rho_{st} = p(st)|\psi_w \rangle \langle \psi_w|$ up to local isometries, where
\begin{equation}
| \psi_w \rangle = \cos \theta_w |00\rangle + \sin \theta_w | 11 \rangle 
\label{eq:tilted-Hardy-state}
\end{equation}
and $(\sin 2\theta_w - 3)^2 = 4w + 5$ for each $-1/4 < w < 1$.
\end{proposition}

These conditions extend the tilted Hardy tests proposed in \cite{ZRLH23} to the scenario with untrusted sources. Notice that every partially entangled pure state of two qubits can be obtained from Eq.~\eqref{eq:tilted-Hardy-state} for a suitable $-1/4 < w < 1$. Moreover, the maximal quantum violation is explicitly given by
\begin{equation}
q(w) =[(4w+5)^{3/2} - (12w+11)] /(2w+2) \, ,
\end{equation}
and the conditions also self-test the measurements to be a suitable mixture of Pauli measurements \cite[Proposition 1]{ZRLH23}. 

\begin{proof}
Imagine $p(stab|xy) = \state{st}{A_{a|x} \otimes B_{b|y}}$ is a quantum realization satisfying the conditions in Eqs.~\eqref{eq:HT-zero-1}--\eqref{eq:HT-violation}. Lemma~\ref{lemma:impossible} implies that the impossible events are independent of the source. Therefore, for each $s,t$,
\begin{align}
    \state{st}{A_{0|0}\otimes B_{1|1}} & = 0 \, , \label{eq:st-zero-1} \\
    \state{st}{A_{1|1}\otimes B_{0|0}} & = 0 \, , \\
    \state{st}{A_{0|1}\otimes B_{0|1}} & = 0 \, .
    \label{eq:st-zero-3}
\end{align}
Moreover, Eq.~\eqref{eq:HT-violation} provides
\begin{equation}
 \state{00}{A_{0|0} \otimes B_{0|0} + w A_{1|0} \otimes B_{1|0}} = p(00)q(w) \, .
\end{equation}
These conditions already self-test $\rho_{00} = p(00)|\psi_w \rangle \langle \psi_w|$ (and additionally the action of the measurements $A_{a|x}$ and $B_{b|y}$) up to local isometries \cite[Proposition 1]{ZRLH23}. 

To conclude the proof, we show that Eqs.~\eqref{eq:st-zero-1}--\eqref{eq:st-zero-3} imply
\begin{equation}
 \state{st}{A_{0|0} \otimes B_{0|0} + w A_{1|0} \otimes B_{1|0}} = p(st) q(w) \, ,
\end{equation}
which again self-test $\rho_{st} = p(st) |\psi_w \rangle \langle \psi_w |$ for each $s,t$.
Indeed, since the measurements are dichotomic, Jordan lemma provides block-diagonal decompositions into two-dimensional blocks $A_{a|x} = \bigoplus_i A_{a|x}^i$ and $B_{b|y} = \bigoplus_j B_{b|y}^j$. We denote $\rho_{st}^{ij}$ the restriction of the state $\rho_{st}$ to the block where $A^i_{a|x} \otimes B^j_{b|y}$ acts, and $\lambda_{st}^{ij} = \state{st,ij}{\id}$ its normalization. Then,
\begin{align}
    \state{st, ij}{A^i_{0|0}\otimes B^j_{1|1}} & = 0 \, , \label{eq:st-ij-zero-1} \\
    \state{st, ij}{A^i_{1|1}\otimes B^j_{0|0}} & = 0 \, , \label{eq:st-ij-zero-2} \\
    \state{st, ij}{A^i_{0|1}\otimes B^j_{0|1}} & = 0 \, , \label{eq:st-ij-zero-3} \\
    \state{00, ij}{A^i_{0|0}\otimes B^j_{0|0} + w A^i_{1|0}\otimes B^j_{1|0}} & = \lambda_{00}^{ij} q(w) \, .
\end{align}
The last condition follows from the maximality of $q(w)$.
These correlations self-test each block $\rho^{00}_{ij} = \lambda^{00}_{ij} |\psi_w \rangle \langle \psi_w |$, and the measurements $A_{a|x} \otimes B_{b|y}$ up to local unitaries. Now, the three linearly independent zeroes in Eqs.~\eqref{eq:st-ij-zero-1}--\eqref{eq:st-ij-zero-3} have a one-dimensional feasible subspace in each four-dimensional block, which completely determine $\rho_{st}^{ij} = \lambda_{st}^{ij} |\psi_w \rangle \langle \psi_w |$. Then, $\state{st}{A_{0|0}\otimes B_{0|0} + w A_{1|0}\otimes B_{1|0}} = \sum_{ij} \lambda^{st}_{ij} q(w) = p(st) q(w)$ for each $s,t$.
\end{proof}

Remarkably, protocols based on impossibility generally do not work with maximal entanglement. Indeed, pseudo-telepathy games with maximal entanglement correspond with Kochen-Specker proofs of contextuality \cite{RW04}, which do not exist in dimension two.
Moreover, nonlocality based on possibilities in Bell scenarios with two settings or two outcomes automatically implies a Hardy test \cite{MF12}, which is incompatible with maximal entanglement in dimension two. Self-testing of maximal entanglement with pseudo-telepathy games is only known for specific dimensions: with the magic square in dimension $4$ \cite{WBMS16} or the magic pentagram in dimension $8$ \cite{kalev2018rigidity}.

\textit{Self-testing qudit states with untrusted sources.-}
Self-testing of pure bipartite partially entangled states with untrusted sources extends to arbitrary dimension. For this, we first extend the self-testing with Hardy tests to arbitrary dimension. Our argument builds upon the pioneering results~\cite{YN12, CGS17} to combine partial self-tests for two-dimensional substates from the previous section into a consistent self-test for the full state.

\begin{figure}[h]
    \centering  \includegraphics[width=0.8\linewidth]{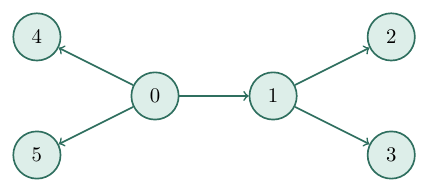}
    \caption{{\bfseries Description of the self-testing protocol} for the state $\ket{\psi} = 1/6 \ket{00} + 1/8 \ket{11} + 1/6 \ket{22} + 1/6\ket{33} + 1/8\ket{44} + 1/4 \ket{55}$. Take the inhomogeneous covering tree rooted at $0$ with edges $(0,1), (0,4), (0,5), (1,2)$ and $(1,3)$. Each party has one measurement with six outcomes $A_0$ and $B_0$. For each edge $(0_i, 1_i)$, say $(0, 4)$, each party has one dichotomic measurement $A_{i,1}$ and $B_{i, 1}$ and the virtual dichotomic measurement $A_{i,0} = (A_{0|0}, A_{4|0})$ and $B_{i,0} = (B_{0|0}, B_{4,|0})$ that is used for a Hardy self-test of the substate $|\psi_i \rangle = c_0|00\rangle + c_4 |44 \rangle$. Namely, the distribution $p_i(ab|xy)$ for these measurements satisfies $p_i(01|01) = p_i(10|10) = p_i(00|11) = 0$ and $p_i(00|00) + w_i p_i(11|00) = p_i q(w_i)$. Here, $p_i = (c_0^2 + c_4^2)$ and $w_i$ is the parameter that self-tests the state with angle $\operatorname{arctan}(c_4/c_0)$.
    }
    \label{fig:covering-tree}
\end{figure}

We say $E = (0_i, 1_i)_{i < d}$ is an \emph{inhomogeneous covering tree} for the state
$\ket{\psi} = \sum_{i<d} c_{i}\ket{ii}$ when $c_{0_i} \neq c_{1_i}$ and every $k < d$ appears in some edge $(0_i, 1_i)$. We say the tree is \emph{rooted} at $k$ if for each $k'$ there are $(k, k_1), \ldots, (k_n, k')\in E$. Every non-maximally entangled state admits a rooted inhomogeneous covering tree: if $c_0 \neq c_1$ take $E = \{(0,k)\}_{c_k \neq c_0} \cup \{(1, k)\}_{c_k = c_0}$ (see Figure~\ref{fig:covering-tree}). Hardy tests over the edges of the covering tree self-test $\ket{\psi}$ up to local isometries. Consider local measurements $A_0$ and $B_0$ with $d$ outcomes, which define virtual local dichotomic measurements $A_{i, 0} = (A_{0_i|0}, A_{1_i|0})$ and $B_{i,0} = (B_{0_i|0}, B_{1_i|0})$ for each $(0_i, 1_i)\in E$. Additionally consider local dichotomic measurements $A_{i, 1}$ and $ B_{i, 1}$ for each $(0_i, 1_i)\in E$. We write $p(ab|x_iy_j) = \state{}{A_{i,a|x} \otimes B_{j, b|y}}$ for the correlations in this Bell experiment with one measurement with $d$ outcomes and $d$ dichotomic measurements. 

Recall the following result \cite[Lemma 2]{CGS17} (c.f. \cite[\S A4a]{YN12}):
\begin{lemma}
    \label{lemma:isometry}
    Let $\ket{\psi} = \sum_{k < d} c_k \ket{kk}$ with $c_k > 0$  be a state and $P_A$ and $P_B$ projective measurements with $d$ outcomes. If there are unitaries $X_A^k$ and $X_B^k$ for each $k < d$ such that
    \begin{align}
        & P_A^k \ket{\psi '} = P_B^k \ket{\psi'} \, , \\
        & X_A^k X_B^k P_B^k \ket{\psi '} = (c_k/c_0) P_A^0 \ket{\psi '} \, , 
    \end{align}
    then $\Phi \ket{\psi'}\otimes \ket{00} = \ket \xi \otimes \ket{\psi}$ for some local isometry $\Phi$.
\end{lemma}

For each $(0_i, 1_i) \in E$ denote $p_i(ab|xy) = p(ab|x_i y_i)$ the (subnormalized) correlations in the Bell experiments with dichotomic measurements $A_{i,0}, A_{i,1}, B_{i,0}, B_{i,1}$.

\begin{proposition}
\label{prop:self-test-qudit-hardy}
Assume that the correlations $p_i(ab|xy)$ obey
\begin{align}
        & p_i = \sum_{ab} p_i(ab|xy) = c_{0_i}^2 + c_{1_i}^2 \, , \\
        & p_i(01|01) = 0 \, , \\
        & p_i(10|10) = 0 \, , \\
        & p_i(00|11) = 0 \, , \\
        & p_i(00|00) + w_i p_i(11|00) = p_i q(w_i) \, .
    \end{align}
    Here, $q(w_i)$ is the maximal quantum value compatible with the zeroes of the tilted Hardy expression with parameter $w_i$, for $(\sin 2\theta_i - 3)^2 = 4w_i + 5$ and $\tan \theta_i = c_{1_i}/c_{0_i}$. Then, $\rho = |\psi \rangle \langle \psi|$ up to local isometries.
    \end{proposition}

    \begin{proof}
    Imagine that $\ket{\psi'}$ and $A_{i,0}, A_{i, 1}, B_{i, 0}, B_{i, 1}$ provide a quantum realization. As previously argued, these conditions force that $\ket{\psi'_{ii}} = \sum_{nm} \sqrt{\state{ii,nm}{\id_{ii}}} \ket{\psi_{ii}}_{nm}$ for a certain decomposition of the space into two-dimensional blocks, where $\sqrt {p_i} \ket{\psi_{ii}}_{nm} = c_{0_i}\ket{0 0}_{nm} + c_{1_i} \ket{1 1}_{nm}$ and $\ket{\psi'_{ii}}$ is the component of $\ket{\psi'}$ in the space $H_{0_i} \oplus H_{1_i}$. For $k < d$ take $(0_i, 1_i) \in E$ such that $a_i =k$. For $P_A^k = A_{a_i | 0}$ and $P_B^k = B_{a_i|0}$ we find that
    \begin{align}
        A_{a_i|0} \ket{\psi'} =  \sum_{nm}\sqrt{\state{ii,nm}{\id}} c_{a_i} \ket{aa}_{nm} = B_{a_i|0} \ket{\psi'} \, .
    \end{align}
    Where we have used $\state{ii,nm}{\id_{ii}} = \state{ii,nm}{\id} p_{i}$. For each $(0_i, 1_i) \in E$ denote $U_A^i$ and $U_B^i$ the local unitaries flipping each block in the decomposition above, $\ket{a}_n \mapsto \ket{\bar a}_n $. Now let $0$ be the root of $E$. For $k < d$ take (the unique) $(0, k_1), \ldots, (k_t, k) \in E$ and define $X_A^k = U_A^0 \ldots U_A^t$.
    Notice
    \begin{align}
        U_A^i U_B^i P_B^{a_i} \ket {\psi'} =  \sum_{nm}\sqrt{\state{ii,nm}{\id}} c_{a_i} \ket{\bar a\bar a}_{nm} = \frac{c_{a_i}}{c_{\bar a_i}}P_A^{\bar a_i} \ket{\psi'} \, .
    \end{align}
    Therefore,
    \begin{align}
        X_A^k X_B^k P_B^{k} \ket {\psi'} = \frac{c_{k_1}}{c_{0}} \ldots \frac{c_{k}}{c_{k_t}}P_A^{0} \ket{\psi'} = \frac{c_{k}}{c_{0}} P_A^{0} \ket{\psi'}\, .
    \end{align}
    Lemma~\ref{lemma:isometry} provides the isometry $\Phi\ket{\psi'} \otimes \ket{00} = \ket{\xi} \otimes \ket{\psi}$.
    \end{proof}

Now consider the scenario with untrusted sources with residual randomness. That is, there is a classical-quantum state $\rho = \sum_{st}|st\rangle \langle st | \otimes \rho_{st}$ correlating the sources and device. For each $(0_i, 1_i)\in E$ there are now distributions $p_i(stab|xy) = \state{st}{A_{a|x_i} \otimes B_{b|y_i}}$ for each $s,t$. Conditions in Proposition~\ref{prop:self-test-qudit-hardy} for $p(st)p_i(ab|st) = p_i(stab|st)$ still self-test the state of the device up to local isometries.

\begin{proposition}
\label{prop:self-test-qudit-hardy-source}
Assume that the correlations $p_i(stab|xy)$ obey
\begin{align}
        & p_i(st) = \sum_{ab} p(st)p_i(stab|st) = p(st)(c_{0_i}^2 + c_{1_i}^2) \, , \label{eq:HT-us-norm}\\
        & p_i(0101|01) = 0 \, , \label{eq:HT-us-zero-1} \\
        & p_i(1010|10) = 0 \, , \label{eq:HT-us-zero-2} \\
        & p_i(1100|11) = 0 \, , \label{eq:HT-us-zero-3} \\
        & p_i(0000|00) + w_i p_i(0011|00) = p_i(00) q(w_i) \, . \label{eq:HT-us-violation}
    \end{align}
    Here, $q(w_i)$ is the maximal quantum value of the tilted Hardy test with parameter $w_i$, for $(\sin 2\theta_i - 3)^2 = 4w_i + 5$ and $\tan \theta_i = c_{1_i}/c_{0_i}$. These conditions can only be satisfied when $\rho_{st} = p(st)|\psi \rangle \langle \psi|$ up to local isometries for each $s,t$.
    \end{proposition}

    The proof follows the argument in Proposition~\ref{prop:selftesting-2qub-MI}, now using the generalized Hardy self-tests from Proposition~\ref{prop:self-test-qudit-hardy}.
    \begin{proof}
    Imagine that $p_i(stab|xy) = \state{st}{A_{a|x_i} \otimes B_{b|y_i}}$ is a quantum realization satisfying the conditions in Eqs.~\eqref{eq:HT-us-norm}--\eqref{eq:HT-us-violation}. Lemma~\ref{lemma:impossible} implies that the impossible events are independent of the source. Therefore, for each $s,t$,
        \begin{align}
            \state{st}{A_{0|0_i} \otimes B_{1|1_i}} & = 0 \, , \\
            \state{st}{A_{1|1_i} \otimes B_{0|0_i}} & = 0 \, , \\
            \state{st}{A_{0|1_i} \otimes B_{0|1_i}} & = 0 \, .
        \end{align}
        Moreover, Eq.~\eqref{eq:HT-us-violation} provides
    \begin{equation}
     \state{00}{A_{0|0_i} \otimes B_{0|0_i} + w_i A_{1|0_i} \otimes B_{1|0_i}} = p_i(00) q(w_i) \, .
    \end{equation}
    It follows from Proposition~\ref{prop:self-test-qudit-hardy} that $\rho_{00} = p(00)|\psi \rangle \langle \psi |$ up to local isometries. The dichotomic measurements $(A_{0_i}, A_{1_i})$ and $(B_{0_i}, B_{1_i})$ admit block-diagonal decompositions. Since there is a one dimensional space satisfying the three zeros in Eqs.~\eqref{eq:HT-us-zero-1}--\eqref{eq:HT-us-zero-3} for each four dimensional block, the marginal state on each block must be proportional to the corresponding block of $\rho_{00}$, which achieves the maximal value $q(w_i)$ up to the normalization of $\rho_{st}$. Therefore, for each $s,t$,
    \begin{equation}
     \state{st}{A_{0|0_i} \otimes B_{0|0_i} + w_i A_{1|0_i} \otimes B_{1|0_i}} = p_i(st) q(w_i) \, .
    \end{equation}
    It then follows from Proposition~\ref{prop:self-test-qudit-hardy} that $\rho_{st} = p(st) |\psi \rangle \langle \psi |$ up to local isometries for each $s,t$.
    \end{proof}

    We can summarize our main contribution as follows.

    \begin{theorem}
        Let $\ket{\psi} = \sum_{i<d} c_{i}\ket{ii}$ with $c_i > 0$ be a pure partially entangled state. There are conditions on a Bell scenario with one measurement with $d$ outcomes and $d-1$ measurements with $2$ outcomes per party that self-test $\ket{\psi}$ up to local isometries for any untrusted source with residual randomness.
    \end{theorem}
    
Additional structure on the covering tree $E$ may reduce the number of measurement settings in the self-testing protocol. Notice that sets of dichotomic measurements $A_{1_{i_1}}, \ldots, A_{1_{i_k}}$ and $B_{1_{i_1}}, \ldots, B_{1_{i_k}}$ corresponding to non-overlapping edges $(0_{i_1}, 1_{i_1}), \ldots, (0_{i_k}, 1_{i_k}) \in E$ can be obtained from single measurements $A_{1}$ and $B_{1}$ with $2k$ outcomes. In the generic case where all Schmidt coefficients of $\ket{\psi}$ are distinct, two such sets of measurements with $d$ outcomes are enough, as in \cite[Lemma 2]{CGS17} (c.f. \cite[\S A4a]{YN12}). 

\textit{Self-testing with Bell values.-}
The presence of untrusted sources in Bell experiments manifests a fundamental difference between Bell inequalities and Hardy tests already in the scenario with two settings and two outcomes: while Hardy self-tests succeed with arbitrary residual randomness $l \leq p(st|abxy) \leq u$ for any $0 < l \leq u < 1$ (not full dependence), Bell values alone fail to self-test any state for $l < 1/4 < u$ (beyond the independence assumption) .

To see this, notice that a CHSH value $2\sqrt 2$ does not self-test in the scenario with untrusted sources:
\begin{equation}
    \state{00}{A_0 \otimes B_0} + \state{01}{A_0 \otimes B_1} + \state{10}{A_1 \otimes B_0} - 
    \state{11}{A_1 \otimes B_1} \, .
    \label{eq:CHSH-sources}
\end{equation}
Take measurements $A_0 = X$, $A_1 = Z$, $B_0 = (X + Z)/\sqrt 2$ and $B_1 = (X - Z)/\sqrt 2$, and states $\rho_{00} = \rho_{01}$ and $\rho_{10} = \rho_{11}$ with normalization $\state{st}{\id} = 1/4$. The value of Eq.~\eqref{eq:CHSH-sources} is 
$(\state{00}{X\otimes X} + \state{11}{Z\otimes Z})\sqrt 2$, which achieves $1/\sqrt 2$ (the normalized CHSH value) if and only if $\state{00}{X \otimes X} = 1$ and $\state{11}{Z\otimes Z} = 1$. In particular, these conditions are satisfied with the maximally entangled state, for which there are no impossible events, $p(stab|xy) = (2\pm \sqrt 2)/8$. Therefore, there exist by continuity $\rho_{00} \neq \rho_{11}$ achieving the maximal value with residual randomness $l \leq p(st|abxy) \leq u$ for arbitrary $l < 1/4 < u$ (beyond the independence assumption). A CHSH value above $2\sqrt 2$ 
can only be obtained when the subnormalized states $\rho_{st}$ are not identical, which already detects correlations with the source. 

\textit{Conclusions.-} Self-testing of any bipartite pure partially entangled state is possible beyond the measurement independence assumption. In particular, in realistic scenarios in which the settings in a Bell experiment are chosen with randomness sources, under the assumption that the behaviour of the source cannot be predicted from statistics on the device. 
Remarkably, high degeneracies in maximally entangled states turn out to be a burden for self-testing with possibilistic tests, rather than an advantage as with probabilistic tests in the measurement independent scenario. Maximal entanglement in dimension two cannot be self-tested with Hardy tests. While maximal entanglement in higher dimensions does exhibit possibilistic nonlocality, nonlocality alone does not guarantee self-testing: apart from specific cases where self-testing is possible with magic square or magic pentagram games, the general question remains open. Another natural question that follows from our work is whether robust self-testing is possible with untrusted randomness sources. These questions and a detailed comparison with the scenario of measurement dependent locality will be addressed in forthcoming work.

\textit{Acknowledgments.-}
We thank Yuan Liu, Liu Mingxuan, Stefano Pironio, Valerio Scarani and Peter Sidajaya for meaningful discussions.
We acknowledge support from the General Research Fund (GRF) Grant No.\ 17211122, and the Research Impact Fund (RIF) Grant No.\ R7035-21.	

\bibliographystyle{apsrev4-2}
\bibliography{references}

\onecolumngrid

\appendix
\newpage

\end{document}